\documentclass[a4paper,11pt]{article}

\usepackage{preamble-fv-simple}
\usepackage{macros}

\usepackage[T1]{fontenc}

\title{Diverse Pairs of Matchings}

\author[1]{Fedor~V.~Fomin}
\author[1]{Petr~A.~Golovach}
\author[1]{Lars~Jaffke}
\author[2,3]{Geevarghese~Philip}
\author[4,5]{Danil~Sagunov}

\affil[1]{University of Bergen, Bergen, Norway}
\affil[ ]{\texttt{\{fedor.fomin,petr.golovach,lars.jaffke\}@uib.no}}
\affil[2]{Chennai Mathematical Institute, Chennai, India}
\affil[3]{UMI ReLaX}
\affil[ ]{\texttt{gphilip@cmi.ac.in}}
\affil[4]{St.\ Petersburg Department of V.A.\ Steklov Institute of Mathematics, St.\ Petersburg, Russia}
\affil[5]{JetBrains Research, St.\ Petersburg, Russia}
\affil[ ]{\texttt{danilka.pro@gmail.com}}

\newenvironment{nestedclaimthm}
	{\begin{nestedclaim}}
	{\end{nestedclaim}}
	
\newenvironment{nestedobservationthm}
	{\begin{nestedobservation}}
	{\end{nestedobservation}}

\begin{document}

\maketitle

\begin{abstract}
We initiate the study of the \textsc{Diverse Pair of (Maximum/ Perfect) Matchings} problems 
which given a graph $G$ and an integer $k$, ask whether $G$ has two (maximum/perfect) matchings
whose symmetric difference is at least $k$.
\textsc{Diverse Pair of Matchings} 
(asking for two not necessarily maximum or perfect matchings) 
is \NP-complete on general graphs 
if $k$ is part of the input, and we consider two restricted variants.
First, we show that on bipartite graphs, the problem is polynomial-time solvable,
and second we show that \textsc{Diverse Pair of Maximum Matchings}
is \FPT parameterized by $k$.
We round off the work by showing that \textsc{Diverse Pair of Matchings} 
has a kernel on $\calO(k^2)$ vertices.
\end{abstract}


\section{Introduction} 
Matching is one of the most fundamental notions in graph theory whose study can be traced back to the classical theorems of  K\H{o}nig~\cite{Konig1916} and Hall~\cite{Hall1935}. 
The first chapter of the book of Lov{\'a}sz and Plummer
\cite{LP09} devoted to matching contains a nice historical overview on the development of the matching problem. The problem of finding a maximum size or a perfect matching are the classical algorithmic problems; an incomplete list of references covering the history of algorithmic improvements on these problems is \cite{edmonds1965paths,HopcroftK73,Karzanov74,Lovasz79,MuchaS04,RabinV89,vazirani2014proof,Madry13}, see also the book of Schrijver  \cite{Schrijver03} for a historical overview of matching algorithms.

In this paper we initiate the algorithmic study of the \emph{diverse} matching
problem.  In this problem, we are to find a pair of matchings which are
different from each other as much as possible.  More formally, we want the size
of their symmetric difference to be large. Recall that the 
\emph{symmetric difference} of two sets \(X, Y\) is defined as
\[
  X \symdiff Y = (X \setminus Y) \cup (Y \setminus X).
\]

We study the following problem.
\fancyproblemdef
	{Diverse Pair of (Maximum/Perfect) Matchings}
	{Graph $G$, integer $k$}
	{Does $G$ contain two (maximum/perfect) matchings $M_1$, $M_2$ such that $\card{M_1 \symdiff M_2} \ge k$?}

    Diversity-enhancing is one of the key goals in developing professional
    social matching systems \cite{DBLP:journals/cacm/0002HK20}.  For example,
    consider the problem of assigning agents to perform various tasks (say, bus
    drivers to bus routes or cleaners to different locations).  To avoid
    monotony, which is one of the declared enemies of happiness at work, the
    practice is to reassign agents to new tasks. 
    In this case, we would be very
    much interested in designing a schedule with diverse assignments. 
    To give another illustration, assume that a teacher should give a series of
    assignments to students that are expected to work in pairs.  From one side,
    the teacher wishes to follow the preferences of the students given by a
    graph, but from the other side, it is preferable to facilitate collaboration
    between different students. This leads to the problem of finding diverse
    perfect matchings in the preference graph.

    We now briefly motivate why finding a diverse set of maximum/perfect
    matchings in a graph would be of interest.  From a graph-theoretic point of
    view, in the simplest model, one maximum/perfect matching is as good as the
    other.  But in a practical setting this is rarely the case since there is a
    large amount of \emph{side information} that determines how an assignment
    (for instance agents to tasks) is received.
    Some side information is modeled by maximum weight matchings, or via notions
    from social choice theory such as stable or envy free
    matchings~\cite{BrandtEtAl2016}.  Nevertheless, this approach has its
    natural limitations; some side information may complicate the model,
    rendering it intractable, while some side information may even be impossible
    to include in a model.

    For instance, if we allow agents to have incomplete preference lists or
    ties, then the corresponding maximum stable matching problem is \NP-hard,
    even in severely restricted cases~\cite{OMalleyThesis}.
    Other side information may be a priori unknown, and only once presented with
    a number of alternatives, we may be able to decide which assignment is the
    most desirable.
    In that case it is key that the presented alternatives are diverse,
    otherwise the insight we gain is comparable to that of having a single fixed
    assignment and is therefore negligible.
    Similar motivations for finding diverse solution sets in combinatorial
    problems can be found in~\cite{BasteEtAl2019,BasteEtAl2019b}.

\subparagraph*{Our results and methods.} 
While a perfect or a maximum matching in a graph can be found in polynomial time, 
this is not true anymore for the diverse variant of the problem, even in graphs of maximum degree three.
Matching problems are often considered on bipartite graphs, and we show that
\textsc{Diverse Pair of Maximum Matchings} remains polynomial-time solvable in this case.

The intractability of the problem in the general case also suggests to look at it from the 
perspective of parameterized complexity~\cite{DowneyFellowsBook,CyganEtAl15} 
and kernelization~\cite{kernelizationbook19}. 
We show that the problem is \FPT parameterized by $k$, 
by giving a randomized $4^k\cdot n^{\calO(1)}$ time algorithm,
and we give a derandomized version of this algorithm that runs in time $4^kk^{\calO(\log k)}\cdot n^{\calO(1)}$.
Finally, we show that the problem asking for a diverse pair of (not necessarily maximum) matchings
admits a kernel on $\calO(k^2)$ vertices.

The randomized algorithm for \divmatchmax is obtained via a combination of color-coding~\cite{AlonYZ95}
and the polynomial-solvability of finding a \emph{minimum cost} maximum matching in a graph~\cite{GabowTarjan1991}.
We derandomize this algorithm via universal sets~\cite{NSS94}.
The kernelization algorithm for \divmatch first finds a maximal matching $M$ in the graph.
If $M$ is large enough, then we can conclude that we are dealing with a \yes-instance by splitting
$M$ into two matchings. 
Otherwise, the endpoints of $M$ form a vertex cover of the input graph
which allows us to shrink the graph without changing the answer to the problem.

\subparagraph*{Related work.}
A well-studied generalization of matchings in graphs is that of a \emph{$b$-matching},
where $b$ is an integer; see for instance~\cite{LP09}.
Given a graph $G$ and an integer $b$,
a \emph{$b$-matching} is an assignment of an integer $\mu(e)$ to each edge $e$ of $G$,
such that for each vertex $v$, the sum over all its incident edges $e_v$ of $\mu(e_v)$ is bounded by $b$.
The size of a $b$-matching is the sum over all edges $e$ in $G$ of $\mu(e)$.
The $1$-matchings of a graph precisely correspond its matchings
(via the edges $e$ with $\mu(e) = 1$).
However, a $2$-matching is not always the union of two matchings: 
take for instance a triangle.
Then, assigning a value of $1$ to all its edges gives a $2$-matching; 
while any matching can have at most one edge from a triangle.
Therefore, finding diverse pairs of matchings is not the same as finding $2$-matchings.

Finding $q$ pairwise \emph{disjoint} matchings of large total size corresponds to finding
large subgraphs that can be $q$-edge colored, each matching constitutes a color class.
The \textsc{Maximum $q$-Edge Colorable Subgraph} problem asks for the largest edge-subgraph 
that can be properly colored with $q$ colors.
This problem is known to be hard to approximate~\cite{FeigeEtAl2002}.

Let $G$ be a graph with edge set $E$ and maximum degree $\Delta$.
Any proper edge coloring requires at least $\Delta$ colors. 
On the other hand,
Vizing's Theorem~\cite{VizingTheorem} asserts that every graph can be properly edge-colored with $\Delta + 1$ colors.
A consequence of this result is that (any) graph $G$ contains a $\Delta$-colorable subgraph with 
at least $\frac{\Delta}{\Delta+1}\card{E}$ edges;
which is tight when $\Delta$ is even as witnessed by the complete graph $K_{\Delta+1}$.
This motivated research in improving the lower bound when $\Delta$ is odd, 
or when $\Delta$ is even and $K_{\Delta+1}$ is excluded.
Kami\'{n}ski and Kowalik~\cite{KK2014} gave several improved lower bounds for the cases when $\Delta \le 7$.

The difference with \textsc{Maximum $2$-Edge Colorable Subgraph} is that in \divmatchmax,
we require matchings (or: color classes) to be of maximum size, while in the former problem,
we only want to maximize the total number of edges in the two color classes.

A recent manuscript due to Fellows~\cite{FellowsDiverse} 
initiated the study of finding diverse sets of solutions to \NP-hard combinatorial problems 
from the viewpoint of parameterized complexity~\cite{BasteEtAl2019,BasteEtAl2019b}.
Concretely, Baste et al.~\cite{BasteEtAl2019} showed that a large class of vertex subset problems
that are \FPT parameterized by treewidth have \FPT algorithms in their diverse variant, 
parameterized by treewidth plus the number of requested solutions.
Moreover, Baste et al.~\cite{BasteEtAl2019b} showed analogous results for hitting set problems 
parameterized by solution size plus number of requested solutions.
Our work contrasts this in that
the classical variant of the problem we consider is polynomial-time solvable,
while its diverse variant becomes \NP-hard, even when asking for only two solutions.

Very recently, Hanaka et al.~\cite{HanakaEtAl2020} gave efficient algorithms
for finding diverse sets of solutions to several other combinatorial problems.
This includes an \FPT-algorithm for finding diverse sets of matchings in a graph.
However, their result is different from ours.
We give an \FPT-algorithm for finding a diverse pair of \emph{maximum} or \emph{perfect} matchings,
and our parameter is the size of the symmetric difference between the matchings, in other words, the diversity measure.
In~\cite{HanakaEtAl2020}, the parameter is the size of the matchings plus the number of requested solutions,
and the matchings do not need to be of maximum cardinality.
Note that in this setting, the maximum possible diversity is bounded in terms of the parameter as well.
In the case that we drop the maximum cardinality requirement on the matchings, 
we even obtained a \emph{polynomial kernel} for finding a diverse pair of matchings.
By the same arguments given in the proof of Theorem~\ref{thm:kernel},
we can derive that the problem of finding a diverse set of $r$ matchings of size $k$ parameterized by $k + r$ 
considered in~\cite{HanakaEtAl2020}
is not only \FPT but has a polynomial kernel.

%
%

%
%
%
%


\section{Preliminaries}
We assume the reader to be familiar with basic notions in graph theory and 
parameterized complexity and refer to~\cite{diestel} 
and~\cite{DowneyFellowsBook,CyganEtAl15,kernelizationbook19}, respectively,
for the necessary background.

All graphs considered in this work are finite,
undirected, simple, and without self-loops.
For a graph $G$ we denote by $V(G)$ its set of \emph{vertices} and by $E(G)$ its set of \emph{edges}.
For an edge $uv \in E(G)$, we call $u$ and $v$ its \emph{endpoints}.
For a vertex $v$ of a graph $G$, $N_G(v) \defeq \{w \in V(G) \mid vw \in E(G)\}$ 
is the set of \emph{neighbors} of $v$ in $G$,
and the \emph{degree} of $v$ is $\deg_G(v) \defeq \card{N_G(v)}$.

The \emph{subgraph induced by $X$}, denoted by $G[X]$, is the graph $(X, \{uv\in E(G) \mid u,v\in X\})$.
For a set of edges $F \subseteq E(G)$, we let $G - F \defeq (V(G), E(G) \setminus F)$.

A graph $G$ is called \emph{empty} if $E(G) = \emptyset$. 
A set of vertices $S \subseteq V(G)$ is an \emph{independent set} if $G[S]$ is empty.
A set $S \subseteq V(G)$ is a \emph{vertex cover} if $V(G) \setminus S$ is an independent set.
A graph $G$ is \emph{bipartite} if its vertex set can be 
partitioned into two nonempty independent sets.

\subparagraph*{NP-Completeness.}
We briefly argue the \NP-completeness of \divmatchgen on $3$-regular graphs which was observed in~\cite{Suo10}.
Membership in \NP is clear.
To show \NP-hardness, we reduce from \textsc{$3$-Edge Coloring} on $3$-regular graphs
which is known to be \NP-complete~\cite{Holyer1981}.
Let $G$ be a $3$-regular graph on $n$ vertices (note that this implies that $n$ is even), 
and consider $(G, n)$ as an instance of \divmatch.
Suppose that $G$ has a proper $3$-edge coloring. 
Since $G$ is $3$-regular, all three colors appear on an incident edge of each vertex.
Therefore, a color class is a perfect matching of $G$, and we can take two color classes
as our solution to $(G, n)$.
Conversely, a solution $(M_1, M_2)$ to $(G, n)$ forms two disjoint matchings of size $n/2$ each.
This implies that both $M_1$ and $M_2$ are perfect, and therefore maximum matchings.
Since each vertex in $G$ has degree three, this means that $M_3 \defeq E(G) \setminus (M_1 \cup M_2)$
also forms a perfect matching in $G$, and therefore $(M_1, M_2, M_3)$ is a proper $3$-edge coloring of $G$.
\begin{observation}[\cite{Suo10}]\label{obs:hardness}
	\textsc{Diverse Pair of (Maximum/Perfect) Matchings} is \NP-complete on $3$-regular graphs.
\end{observation}

\section{A Polynomial-Time Algorithm for Bipartite Graphs}
In this section we show that \divmatchmax is solvable in polynomial time on bipartite graphs
via a reduction to the \textsc{$2$-Factor} problem.
\begin{theorem}
	\divmatchmax is polynomial-time solvable on bipartite graphs.
\end{theorem}
\begin{proof}
	Let $(G,k)$ be a given instance of \divmatchmax, where $G$ is bipartite.
	We show how to reduce this instance to an equivalent instance of the problem of finding maximum-weight $2$-factor of a larger graph $G'$.
	A \emph{$2$-factor} of $G'$ is a subgraph of $G'$ in which the degree of each vertex is equal to $2$.
	Equivalently, $2$-factor of $G'$ is a vertex-disjoint cycle cover of $G'$.
	The problem of finding a maximum-weight factor of a graph is well-known to be solvable in polynomial time using the Tutte's reduction to the problem of finding a maximum-weight perfect matching \cite{LP09,tutte1954short}.
	Our graph $G'$ is an edge-weighted graph with parallel edges.
	We note that the algorithm of finding a maximum-weight factor works fine with such graphs.
	
	We also assume that the two parts of $G$ are of equal size.
	If that is not true, introduce isolated vertices to the smaller part of $G$.
	This does not change the matching structure in $G$, so the obtained instance is equivalent to the initial one.
	Denote the number of vertices in each part of $G$ by $n$, so $|V(G)|=2n$.
	
	We now show how to construct $G'$ given $G$.
	The graph $G'$ is defined on the same vertex set as $G$ is, i.e.\ $V(G')=V(G)$.
	For each edge $uv$ of $G$, $G'$ has two parallel edges between $u$ and $v$.
	One of these edges is assigned weight $1$, and the other is assigned weight $0$.
	In other words, edges of $G$ are doubled in $G'$.
	Additionally, for each pair of vertices $u,v$ from distinct parts of $G$ that are not adjacent in $G$, $G'$ has two parallel edges of weight $-n$ between $u$ and $v$.
	Thus, $G'$ is a complete bipartite graph with doubled edges, and weights of these edges depend on what edges are present in $G$.
	This finishes the construction of $G'$.
	
	\begin{nestedclaimthm}\label{cla:two_factor_claim}
		Let $M_1$ and $M_2$ be a pair of maximum matchings in $G$ that maximize the value of $|M_1 \cup M_2|$.
		Then the maximum weight of a $2$-factor of $G'$ equals $|M_1 \cup M_2|-2n\cdot(n- |M_1|)$.
	\end{nestedclaimthm}
	\begin{claimproof}
		We first show that $G'$ has a $2$-factor of weight at least $|M_1\cup M_2|-2n \cdot (n-|M_1|)$.
		Denote this $2$-factor by $F$.
		It is constructed as follows.
		For each edge of $M_1$ take the corresponding edge of weight $1$ in $G'$ into $F$.
		Then, for each edge in $M_2\setminus M_1$ take the corresponding edge of weight $1$ into $F$. 
		For each edge in $M_1 \cap M_2$, take the corresponding edge of weight $0$ in $G'$ into $F$.
		Clearly, $F$ is now of weight $|M_1\cup M_2|$, but it is not yet a $2$-factor of $G'$, unless $M_1$ and $M_2$ are perfect matchings.
		
		There are $n-|M_1|$ vertices in each part of $G$ that are not saturated by $M_1$.
		Take these $2(n-|M_1|)$ vertices and take an arbitrary matching between them in $G'$.
		All edges of this matching are of weight $-n$, otherwise $M_1$ is not maximum in $G$.
		Add the edges of this matching into $F$.
		Repeat the same for $M_2$, i.e.\ take an arbitrary matching in $G'$ between vertices that are not saturated by $M_2$ and add all its edges into $F$.
		The edges of weight $-n$ of the matchings for $M_1$ and $M_2$ may coincide.
		If an edge of weight $-n$ is presented in both matchings, take both its parallel copies into $F$.
		It is easy to see that $F$ is now a $2$-factor of $G'$, as it consists of edges of two perfect matchings between two parts.
		The weight of $F$ is $|M_1\cup M_2|-n(n-|M_1|)-n(n-|M_2|)=|M_1\cup M_2|-2n\cdot (n-|M_1|)$.
		
		It is left to show that $F$ is indeed a maximum-weight $2$-factor of
        $G'$.  To see this, take a maximum-weight factor $F'$ of $G'$ and
        assume that the weight of $F'$ is greater than the weight of $F$.  Note
        that $F'$ consists of $2n$ edges.  As discussed above, $F'$ forms a
        disjoint union of simple cycles on the vertices of $G'$, where each
        vertex belongs to exactly one cycle.  Note that some of these cycles may
        consist of two parallel edges.  Since $G'$ is bipartite, all of these
        cycles have even length.  Color the edges of $F'$ with two colors so
        that no two consecutive edges have the same color on a cycle.  Then the
        edges of the same color form a perfect matching in $G'$.  Denote these
        matchings by $F'_1$ and $F'_2$.  Let $M'_1$ be the set of original edges
        of $G$ which copies are present in $F'_1$.  Let $M'_2$ be the set of
        edges of $G$ obtained analogously from $F'_2$.  Copies of edges in
        $M'_1$ and $M'_2$ have weights $0$ or $1$ in $G'$.  All other edges in
        $F'_1$ and $F'_2$ are of weight $-n$.
		
		Observe that if an edge of $G$ is present in both $M'_1$ and $M'_2$,
        then one of its copies in $F'$ has weight $0$, and the other has weight
        $1$.  Thus, the total weight of $0$- and $1$-weighted edges in $F'$ is
        at most $|M'_1 \cup M'_2|$.  The number of edges of weight $(-n)$ in
        $F'$ is $2n-|M'_1|-|M'_2|$.  Thus, the total weight of $F'$ is at most
        $|M'_1\cup M'_2|-n(2n-(|M'_1|+|M'_2|))$.
		
		We assumed that the weight of $F'$ is greater than the weight of
        $F$. From this we get that
        \((|M'_1\cup M'_2|-|M_1\cup M_2|)-n(2|M_1|-(|M'_1|+|M'_2|))> 0\) holds;
        equivalently, that
        \begin{equation}
          \label{eq:1}
          (|M'_1\cup M'_2|-|M_1\cup M_2|) > n(2|M_1|-(|M'_1|+|M'_2|)).
        \end{equation}

        Recall that $M'_1$ and $M'_2$ are matchings in $G$.  Suppose $M'_1$ and
        $M'_2$ are \emph{maximum} matchings in $G$. Then the right hand side of
        \autoref{eq:1} evaluates to zero, and---by the definition of $M_1$ and
        $M_2$---the left hand side is \emph{at most} zero. Hence \autoref{eq:1}
        does not hold, a contradiction.  So at least one of $M'_1$ and $M'_2$ is
        \emph{not} a maximum matching. Thus we get that
        \begin{equation}
          \label{eq:2}
        |M'_1|+|M'_2| < 2|M_1|
        \end{equation}
        holds; equivalently, that
        \(|M'_{1}|+|M'_{2}| - |M_{1}| < |M_{1}|\) holds.
        By construction we have that the size of any matching in \(G\) is at
        most \(n\). In particular \(|M_{1}| \leq n\), and so we have that 
        \begin{equation}
          \label{eq:3}
          |M'_{1}|+|M'_{2}| - |M_{1}| < n
        \end{equation}
        holds.
        \autoref{eq:2} can be restated as
        \(2|M_{1}| - (|M'_{1}|+|M'_{2}|) > 0\).  Now, 
        \begin{equation}
          \label{eq:4}
          2|M_{1}| - (|M'_{1}|+|M'_{2}|) \geq 1
        \end{equation}
        holds.  Substituting \autoref{eq:4} in \autoref{eq:1} we get that
        \begin{equation}
          \label{eq:5}
          (|M'_1\cup M'_2|-|M_1\cup M_2|) > n
        \end{equation}
        holds. Observe now that $|M'_1|+|M'_2| \geq |M'_1\cup M'_2|$ and
        $ |M_1| \leq |M_1\cup M_2|$ hold.  Substituting these in \autoref{eq:5}
        we get that $((|M'_1|+|M'_2|)-|M_1|) > n$ holds, which contradicts
        \autoref{eq:3}.
	\end{claimproof}
	
    Now let \(M_{1}, M_{2}\) be two arbitrary maximum matchings of \(G\), and
    let \(\mu(G)\) denote the size of a maximum matching of \(G\). Thus
    \(|M_{1}|=|M_{1}|=\mu(G)\). By the definition of symmetric difference we
    have that
    \(|M_{1} \symdiff M_{2}| = |M_{1} \setminus (M_{1} \cap M_{2})| + |M_{2}
    \setminus (M_{1} \cap M_{2})| = |M_{1}| - |M_{1} \cap M_{2}| + |M_{2}| -
    |(M_{1} \cap M_{2})| = 2\mu(G) - 2|M_{1} \cap M_{2}|\). And since
    \(|(M_{1} \cap M_{2})| = |M_{1}| + |M_{2}| - |M_{1} \cup M_{2}| = 2\mu(G) -
    |M_{1} \cup M_{2}|\) we get that
    \(|M_{1} \symdiff M_{2}| = 2\mu(G) - 2(2\mu(G) - |M_{1} \cup M_{2}|) =
    2(|M_{1} \cup M_{2}| - \mu(G))\). Since \(\mu(G)\) is an invariant of graph
    \(G\) this means that the maximum value of \(|M_{1} \symdiff M_{2}|\) is
    attained by exactly those pairs of maximum matchings \(M_{1}, M_{2}\) which
    maximize the value \(|M_{1} \cup M_{2}|\). Further, let
    \(M_{1}^{\star}, M_{2}^{\star}\) be a pair of maximum matchings such that
    \(|M_{1}^{\star} \cup M_{2}^{\star}|\) is the maximum among all pairs of
    maximum matchings. Then we have that the maximum value of
    \(|M_{1} \symdiff M_{2}|\), over all pairs of maximum matchings, equals
    \(2(|M_{1}^{\star} \cup M_{2}^{\star}| - \mu(G))\).

    From \autoref{cla:two_factor_claim} we get that we can compute the value
    \(|M_{1}^{\star} \cup M_{2}^{\star}|\)---though not the matchings
    \(M_{1}^{\star}\) and \(M_{2}^{\star}\)---in polynomial time, by computing
    the weight of a maximum 2-factor in a derived graph. We can find the maximum
    matching size \(\mu(G)\) of \(G\) in polynomial time as well. So we can
    compute the number \(2(|M_{1}^{\star} \cup M_{2}^{\star}| - \mu(G))\) in
    polynomial time. By the arguments in the previous paragraph, checking
    whether \(2(|M_{1}^{\star} \cup M_{2}^{\star}| - \mu(G)) \geq k\) suffices
    to solve the bipartite instance \((G,k)\) of \divmatchmax.
    \end{proof}

\section{FPT-Algorithm for \divmatchmax}
In this section we give an \FPT-algorithm for \divmatchmax parameterized by $k$.
We first give a randomized algorithm based on the color-coding technique of Alon, Yuster and Zwick \cite{AlonYZ95} in Theorem~\ref{thm:fpt},
and then derandomize this algorithm at the cost of a slightly slower runtime 
in Corollary~\ref{cor:fpt:derandom}.
\newcommand\red{\mathrm{red}}
\newcommand\green{\mathrm{green}}
\newcommand\blue{\mathrm{blue}}
\begin{theorem}\label{thm:fpt}
	\divmatchmax parameterized by $k$ is \FPT.
	More precisely, there is a randomized algorithm that in time $4^k \cdot n^{\calO(1)}$
	finds a solution with constant probability, if it exists, and correctly concludes that there is no solution otherwise,
	where $n$ denotes the number of vertices of the input graph.
\end{theorem}
\begin{proof}
	Let $G$ be the graph of the given instance.
	First, we compute a maximum matching $M$ in $G$ in polynomial time~\cite{edmonds1965paths,GabowTarjan1991,MuchaS04}.
	%
	We check if there is a solution using $M$ as one of the two matchings.
	\begin{nestedclaimthm}\label{claim:base}
		Let $G$ be a graph and $M$ a maximum matching of $G$. 
		One can determine in polynomial time whether $G$ has a maximum matching $M'$
		such that $\card{M \symdiff M'} \ge k$, and construct such a matching if it exists.
	\end{nestedclaimthm}
	\begin{claimproof}
		The algorithm is as follows. 
		Let $c \colon E(G) \to \{0, 1\}$ be a cost function of the edges of $G$, defined as
		\begin{align*}
			c(e) \defeq \left\lbrace\begin{array}{ll}
				1, & \mbox{ if } e \in M \\
				0, &\mbox{ otherwise}
			\end{array}\right.
			~~~\forall e \in E(G)
		\end{align*}
		Let $M'$ be a minimum cost maximum matching in $G$ using the cost function $c$.
		Such a matching $M'$ can be found in polynomial time~\cite{GabowTarjan1991}.
		Due to the cost function $c$, a minimum cost maximum matching in $G$ is one 
		that minimizes the number of edges from $M$. 
		Therefore, $M'$ maximizes the symmetric difference with $M$, over all maximum matchings of $G$.
		We verify whether $\card{M \symdiff M'} \ge k$, and if so, return $M'$.
		Otherwise, we correctly conclude that there is no matching 
		satisfying the conditions of the claim.
	\end{claimproof}
	
	Due to Claim~\ref{claim:base}, we may now assume that for each maximum matching $M'$ of $G$,
	$\card{M' \symdiff M} \le k$.
	We will exploit this property to give an algorithm using color coding (see e.g.~\cite[Chapter 5]{CyganEtAl15}).
	We color the edges of $G$ uniformly at random with colors $\red$ and $\blue$.
	For ease of exposition, we also use the notation `$\red$' and `$\blue$' to denote the \emph{set}
	of edges that received color $\red$ and $\blue$, respectively.
	
	Suppose that there is a solution $(M_1, M_2)$.
	We say that a coloring as above is \emph{good for $(M_1, M_2)$}, 
	if the edges in $M_1 \setminus M_2$ and  $M_2 \setminus M_1$ 
	are colored $\red$ and $\blue$, respectively.
	We call an edge coloring \emph{good}, if it is good for some solution.
	To be able to show that trying $4^k$ colorings to achieve constant success probability suffices,
	we bound the size of these sets.
	
	By Claim~\ref{claim:base}, we know that $\card{M \symdiff M_r} \le k$ for all $r \in \{1,2\}$.
	Since $\card{M_1\symdiff M_2}$ is the Hamming distance between sets, by the triangle inequality, 
	\begin{equation*}
	\card{M_1\symdiff M_2}\leq \card{M_1\symdiff M}+\card{M\symdiff M_2}\leq 2k
	\end{equation*}
	and $\card{M_1 \setminus M_2} +\card{M_2 \setminus M_1} \le 2k$.
	This leads to the following observation.
	\begin{nestedobservationthm}\label{obs:coloring}
		Let $G$ be a graph, let $M$ be a maximum matching of $G$, 
		and suppose that for all maximum matchings $M'$ of $G$, $\card{M \symdiff M'} \le k$.
		Suppose the edges of $G $ are colored uniformly at random with colors $\red$ and $\blue$.
		Suppose there is a solution $(M_1, M_2)$. Then, with probability at least $2^{-2k}$,
		the edge coloring is good for $(M_1, M_2)$.
	\end{nestedobservationthm}
	
	Suppose that our instance is a \yes-instance, and that the edges of $G$ are colored with a good coloring.
	We show how to obtain the solution in polynomial time from the edge-colored graph.
	%
	\begin{nestedclaimthm}\label{claim:good:coloring}
		Let $G$ be a graph, $M$ a maximum matching of $G$, 
		and suppose that the edges of $G$ are colored uniformly at random with colors $\red$ and $\blue$.
		There is an algorithm that runs in polynomial time, and
		if the edge-coloring is good, finds two maximum matchings $M_1$ and $M_2$ in $G$ such that $\card{M_1 \symdiff M_2} \ge k$,
		and reports \no otherwise.
	\end{nestedclaimthm}
	\begin{claimproof}
		The idea is similar to the one given in the algorithm of Claim~\ref{claim:base}.
		To find $M_1$, we define the following  cost function $c_1 \colon E(G) \to \{0, 1\}$:
		\begin{align*}
			c_1(e) \defeq \left\lbrace\begin{array}{ll}
				1, &\mbox{ if } e \in \blue \\
				0, &\mbox{ if } e \in \red
			\end{array}\right. ~~~\forall e \in E(G).
		\end{align*}
		Then, we find a minimum-cost maximum matching $M_1$ of $G$ with the cost function $c_1$
		in polynomial time~\cite{LP09}.
		
		Next, to find $M_2$, we consider the cost function $c_2 \colon E(G) \to \{0, 1\}$, where
		\begin{align*}
			c_2(e) \defeq \left\lbrace\begin{array}{ll}
				1, &\mbox{ if } e \in \red \\
				0, &\mbox{ if } e \in \blue
			\end{array}\right. ~~~\forall e \in E(G),
		\end{align*}
		and find a minimum-cost maximum matching $M_2$ of $G$ with cost function $c_2$
		in polynomial time~\cite{LP09}.
		%
		%
		Now, if $\card{M_1 \symdiff M_2} \ge k$, then we return $(M_1, M_2)$, and we say \no, otherwise.
		
		We now argue the correctness of the algorithm in the case that the edge-coloring of $G$ was good.
		In this case, there is a solution $(M_1^*, M_2^*)$
		such that the edges of $M_1^*\setminus M_2^*$ are $\red$ and the edges of $M_2^*\setminus M_1^*$ are $\blue$, and
		$|M_1^*\symdiff M_2^*|\geq k$. We claim that $|M_1\symdiff M_2|\geq |M_1^*\symdiff M_2^*|$.
		
		To obtain a contradiction, assume that  $|M_1\symdiff M_2|<|M_1^*\symdiff M_2^*|$.		 
		
		Since $M_1$, $M_2$, $M_1^*$, and $M_2^*$ are maximum matchings of $G$, they have the same size.
		Therefore, we have that 
		\begin{equation*}
		|M_1\symdiff M_2|+2|M_1\cap M_2|=|M_1|+|M_2|=|M_1^*|+|M_2^*|=	|M_1^*\symdiff M_2^*|+2|M_1\cap M_2|.	
		\end{equation*}
		Since $|M_1\symdiff M_2|<|M_1^*\symdiff M_2^*|$, we obtain that 
		\begin{equation}\label{eq:one}
		|M_1\cap M_2|>|M_1^*\cap M_2^*|. 
		\end{equation}
		Because $M_1^*\setminus M_2^*\subseteq\red$,
		$c_1(M_1^*)=|M_1^*\cap \blue|=|(M_1^*\cap M_2^*)\cap\blue|$ by the definition of the cost function $c_1$. Symmetrically,  $c_2(M_2^*)=|(M_1^*\cap M_2^*)\cap\red|$. Hence,
		\begin{equation}\label{eq:two}		
		c_1(M_1^*)+c_2(M_2^*)=|(M_1^*\cap M_2^*)\cap\blue|+	|(M_1^*\cap M_2^*)\cap\red|=|M_1^*\cap M_2^*|.	
		\end{equation}		
		Notice that $c_1(M_1)=|M_1\cap\blue|\geq |(M_1\cap M_2)\cap\blue|$ and 
		$c_2(M_2)=|M_2\cap\red|\geq |(M_1\cap M_2)\cap\red|$. Therefore, 
		\begin{equation}\label{eq:three}
		c_1(M_1)+c(M_2)\geq |(M_1\cap M_2)\cap\blue|+|(M_1\cap M_2)\cap\red|=|M_1\cap M_2|.		
		\end{equation}
		Combining (\ref{eq:one})--(\ref{eq:three}), we obtain that 
		\begin{equation}\label{eq:four}
		c_1(M_1)+c_2(M_2)\geq |M_1\cap M_2|>|M_1^*\cap M_2^*|=c_1(M_1^*)+c_2(M_2^*).
		\end{equation}
		However, $c_1(M_1)\leq c_1(M_1^*)$ and $c_2(M_2)\leq c_2(M_2^*)$ by the definition of these matchings and 
		\begin{equation*}
		c_1(M_1)+c_2(M_2)\leq c_1(M_1^*)+c_2(M_2^*)
		\end{equation*}
		contradicting (\ref{eq:four}).
		We conclude that $|M_1\symdiff M_2|\geq |M_1^*\symdiff M_2^*|$.		
	\end{claimproof}
	\begin{algorithm}[t]
		\SetKwInOut{Input}{Input}
		\SetKwInOut{Output}{Output}
		\SetKwFor{RepeatTimes}{repeat}{times}{}
		\Input{Graph $G$, integer $k$}
		\Output{If exists, with constant probability, a pair $M_1$, $M_2$ of maximum matchings of $G$ such that $\card{M_1 \symdiff M_2} \ge k$,
			\no otherwise.}
		Compute a maximum matching $M$ of $G$\;\label{alg:fpt:0}
		\If{there is a maximum matching $M'$ in $G$ such that $\card{M \symdiff M'} \ge k$}{\label{alg:fpt:base}
			\Return $(M, M')$\;
		}
		\Else {
			\RepeatTimes{$2^{2k}$}{
				Color the edges of $G$ uniformly at random with colors $\red$ and $\blue$;\\
				run the algorithm of Claim~\ref{claim:good:coloring}\;\label{alg:fpt:good}
				\lIf{the algorithm returned $(M_1, M_2)$}{\Return $(M_1, M_2)$}
			}
			\Return \no\;
		}
		\caption{The algorithm of Theorem~\ref{thm:fpt}.}
		\label{alg:fpt}
	\end{algorithm}
	
	The outline of the procedure is given in Algorithm~\ref{alg:fpt}.
	It is well-known that a maximum matching of a graph can be found in polynomial time,  
	therefore line~\ref{alg:fpt:0} takes polynomial time.
	By Claims~\ref{claim:base} and~\ref{claim:good:coloring}, 
	lines~\ref{alg:fpt:base} and~\ref{alg:fpt:good}, respectively, take polynomial time.
	Moreover, by Observation~\ref{obs:coloring}, if there is a solution, 
	then with probability at least $2^{-2k}$ an edge-coloring as constructed in line~\ref{alg:fpt:good}
	is good, in which case the algorithm finds the solution by Claim~\ref{claim:good:coloring}.
	It is clear that repeating this step $2^{2k}$ times yields a constant success probability.
\end{proof}

The algorithm of Theorem~\ref{thm:fpt} can be derandomized by standard tools (see, e.g.,~\cite[Chapter 5]{CyganEtAl15}).
To do so, we use the following notion of $(\Omega, k)$-universal sets,
which will replace the random coloring step in the above algorithm 
by deterministic choices of colorings. 
\begin{definition}[$(\Omega, k)$-universal set]
	Let $\Omega$ be a set and $k$ be a positive integer with $k \le \card{\Omega}$.
	An \emph{$(\Omega, k)$-universal set} is a family $\calU$ of subsets of $\Omega$ such that
	for any size-$k$ set $S \subseteq \Omega$, the family $\calU_S \defeq \{A \cap S \colon A \in \calU\}$
	contains all subsets of $S$.
\end{definition}

We will use the following construction of a small universal set due to Naor et al.~\cite{NSS94}.
\begin{theorem}[\cite{NSS94}, see also Theorem~5.20 in~\cite{CyganEtAl15}]\label{thm:universal:set}
	For any set $\Omega$ and integer $k \le \card{\Omega}$, one can construct an
	$(\Omega, k)$-universal set of size $2^kk^{\calO(\log k)}\log(\card{\Omega})$
	in time $2^kk^{\calO(\log k)}\card{\Omega} \log(\card{\Omega})$.
\end{theorem}
This immediately gives the following corollary.
\begin{corollary}\label{cor:fpt:derandom}
	There is a deterministic $4^{k}k^{\calO(\log k)} \cdot n^{\calO(1)}$ time algorithm that solves \divmatchmax,
	where $n$ denotes the number of vertices in the input graph.
\end{corollary}
\section{Polynomial kernel for \divmatch}
We now show that the \divmatch problem, asking for a pair of not necessarily maximum matchings
has a kernel on $\calO(k^2)$ vertices. 
Note that the \NP-completeness of this problem is captured 
in Observation~\ref{obs:hardness} as well.
Moreover, we would like to remark that this variant of the problem is only interesting in the case
when the input graph has no matching of size $k$ or more:
otherwise, a maximum matching (which can be found in polynomial time)
forms a trivial solution together with an empty matching.
\begin{theorem}\label{thm:kernel}
	\divmatch parameterized by $k$ has a kernel on $\calO(k^2)$ vertices.
\end{theorem}
\begin{proof}
	Let $(G, k)$ be an instance of \divmatch. 
	We provide a procedure that either correctly concludes that $(G, k)$ is a \yes-instance,
	or marks a set of $\calO(k^2)$ vertices $X \subseteq V(G)$ 
	such that $(G[X], k)$ is equivalent to $(G, k)$.
	
	First, let $M$ be a maximal matching of $G$. 
	If $\card{M} \ge k$, then for any $2$-partition $(M_1, M_2)$ of $M$, we have that 
	$\card{M_1 \symdiff M_2} = \card{M} \ge k$, and therefore $(G, k)$ is a \yes-instance.
	
	Suppose that $\card{M} < k$ and therefore, $\card{V(M)} < 2k$.
	Since $M$ is maximal, $V(M)$ is a vertex cover of $G$, and therefore, $E(G - V(M)) = \emptyset$.
	This motivates the following procedure that produces a set of marked vertices $X \subseteq V(G)$,
	to which we can restrict the instance without changing the answer.
	\begin{enumerate}[label={\arabic*.},ref={\arabic*}]
		\item Initialize $X \defeq V(M)$.
		\item\label{matchings:kernel:marking:2} For each $v \in V(M)$, 
			add a maximal subset of $N_G(v) \setminus V(M)$ of size at most $2k$ to $X$.
	\end{enumerate}
	
	Let $X$ denote the set constructed according to the two previous steps. 
	We show that $(G[X], k)$ is equivalent to $(G, k)$.
	\begin{nestedclaimthm}\label{claim:kernel}
		Let $G$, $k$, $M$, and $X$ be as above. 
		Then, $(G, k)$ is a \yes-instance of \divmatch if and only if 
		$(G[X], k)$ is a \yes-instance of \divmatch.
	\end{nestedclaimthm}
	\begin{claimproof}
		Since $G[X]$ is a subgraph of $G$, it is clear that if $(G[X], k)$ is a \yes-instance, then so is $(G, k)$.
		
		Now suppose that $(G, k)$ is a \yes-instance and let $(M_1, M_2)$ with $\card{M_1 \symdiff M_2} \ge k$ be a solution.
		If $M_1 \cup M_2 \subseteq E(G[X])$, then $(M_1, M_2)$ is also a solution to $(G[X], k)$,
		so suppose that for some $r \in \{1, 2\}$, there is an edge $uv \in M_r$ such that $v \in V(G) \setminus X$.
		Since $V(M) \subseteq X$ and $V(M)$ is a vertex cover of $G$, we may assume that $u \in V(M)$.
		Since $v$ is a neighbor of $u$ in $V(G) \setminus X$, 
		the above marking algorithm added a set of $2k$ neighbors of $u$ in $V(G) \setminus V(M)$ to $X$,
		denote that set by $X_u$. For an illustration see \cref{fig:kernel}.
		\begin{figure}
			\centering
			\includegraphics[height=.15\textheight]{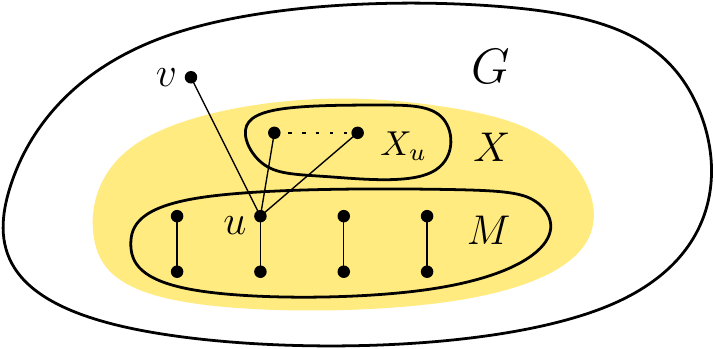}
			\caption{Illustration of the situation in the proof of \cref{claim:kernel}.
				The existence of $v$ implies that $\card{X_u} \ge 2k$,
				and since $V(M)$ is a vertex cover of $G$, the vertices in $X_u$ are pairwise non-adjacent.}
			\label{fig:kernel}
		\end{figure}
		
		Now, since $X_u \subseteq V(G) \setminus V(M)$, and since $V(M)$ is a vertex cover of $G$,
		we have that $E(G[X_u]) = \emptyset$.
		This means in particular that each edge in $M_1 \cup M_2$ has at most one endpoint in $X_u$.
		Therefore, if all vertices in $X_u$ are the endpoint of some edge in either $M_1$ or in $M_2$,
		then $\card{M_1 \cup M_2} \ge 2k$, which implies that at least one of $M_1$ and $M_2$ contains at least $k$ edges.
		Suppose w.l.o.g.\ that $\card{M_1} \ge k$. As above, any $2$-partition $(M_1', M_2')$ of $M_1$ is such that
		$\card{M_1' \symdiff M_2'} = \card{M_1} \ge k$, therefore $(M_1', M_2')$ is a solution to $(G[X], k)$.
		Otherwise, there is a vertex $x \in X_u$ that is \emph{not} the endpoint of any edge in $M_1 \cup M_2$.
		We obtain $M_r^\star$ by removing $uv$ and adding $ux$.
		Then, $(M_r^\star, M_{3-r})$ is still a solution to $(G, k)$, and it uses one more edge in $G[X]$.
		Repeatedly applying this argument shows that $(G[X], k)$ is a \yes-instance.
		%
	\end{claimproof}
	
	The previous claim asserts the correctness of the procedure. Since $\card{V(M)} < 2k$, 
	and for each vertex in $V(M)$, we added at most $2k$ more vertices to $X$, we have that $\card{X} = \calO(k^2)$.
	A maximal matching can be found greedily, and it is clear that the marking procedure runs in polynomial time.
	This yields the result.
\end{proof}

\section{Conclusion}
In this work, we initiated the study of algorithmic problems asking for 
diverse pairs of (maximum/perfect) matchings,
where diverse means that their symmetric difference has to be at least some value $k$.
These problems are \NP-complete on $3$-regular graphs,
and we showed that on bipartite graphs, they become polynomial-time solvable;
while parameterized by $k$, they are \FPT,
and the problem asking for two diverse (not necessarily maximum) matchings
admits a polynomial kernel.

The notion of diverse matchings opens up many natural further research directions.
In this work, we considered the complexity of finding \emph{pairs} of diverse matchings.
What happens when we ask for a larger number of matchings?
In~\cite{BasteEtAl2019,BasteEtAl2019b}, the measure of diversity of a set of solutions is 
the \emph{sum} over all pairs of their symmetric difference.
In this setting, we can obtain an \FPT-algorithm parameterized by the number of requested matchings plus the `diversity target'
using the same approach as in our \FPT-algorithm for \divmatchmax.
However, if we ask for a set of matchings $\calM$ such that \emph{for each pair} $M_1, M_2 \in \calM$,
$\card{M_1 \symdiff M_2} \ge k$, then the situation is much less clear, even asking for three solutions.
Call the corresponding problem \textsc{Diverse Triples of Maximum Matchings}. 
Is it \FPT parameterized by $k$?

While the symmetric difference is a natural measure of diversity of two matchings,
one might consider other measures as well. 
The diversity measure at hand may affect the complexity of the problem,
so it would be interesting to see if there is an (easily computable) diversity measure
under which \divmatchgen becomes $\mathsf{W}$[1]-hard.

\subparagraph*{Acknowledgements.} 
We thank G\"{u}nter Rote for pointing out to us that, given a maximum matching $M$, we can find a maximum matching $M'$ such that $|M\symdiff M'|$ is maximum in polynomial time by the reduction to the \textsc{Minimum Cost Maximum Matching} problem.

\bibliographystyle{plain}
\bibliography{99-ref}

\end{document}